\documentstyle[jair,twoside,11pt,theapa]{article}
\input epsf 

\newcommand{\ignore}[1]{}

\newtheorem{theorem}{Theorem}[section]
\newtheorem{lemma}[theorem]{Lemma}
\newtheorem{corollary}[theorem]{Corollary}

\newtheorem{proposition}[theorem]{Proposition}

\newtheorem{definition}{Definition}
\newtheorem{claim}[theorem]{Claim}

\newcommand{\qed}{\hfill{\fbox{}}}
\newenvironment{proof}{\par\noindent{\bf Proof:}}{\qed \par}

\newcommand{\smallspace}{\vspace{.1in}\noindent}

\newcommand{\cout}[1]{}       

\newcommand{\notation}{\smallspace {\bf Notation:}}
\newcommand{\notations}{\smallspace {\bf Notation:}}

\newcommand {\maximizer} {maximal in its range}
\newcommand {\vcgbased} {VCG-based}

\jairheading{29}{2007}{19--47}{03/06}{05/07}
\ShortHeadings{Computationally Feasible VCG Mechanisms} {Nisan
\& Ronen} \firstpageno{19} 

\begin{document}

\title{Computationally Feasible VCG Mechanisms}

\author{\name Noam Nisan \email noam@cs.huji.ac.il \\
       \addr School of Computer Science and Engineering, \\
       The Hebrew University of Jerusalem, Israel
       \AND
       \name Amir Ronen \email amirr@ie.technion.ac.il \\
       \addr Faculty of Industrial Engineering \& Management, \\
       Technion - Israel Institute of Technology,\\
       Haifa 32000, Israel
       }


\maketitle

\begin{abstract} A major achievement of mechanism design theory is a general
method for the construction of truthful mechanisms called VCG
(Vickrey, Clarke, Groves). When applying this method to complex
problems such as combinatorial auctions, a difficulty arises: VCG
mechanisms are required to compute optimal outcomes and are,
therefore, computationally infeasible. However, if the optimal
outcome is replaced by the results of a sub-optimal algorithm, the
resulting mechanism (termed VCG-based) is no longer necessarily
truthful.
The first part of this paper studies this phenomenon in depth and
shows that it is near universal.  Specifically, we prove that
essentially all reasonable approximations or heuristics for
combinatorial auctions as well as a wide class of cost
minimization problems yield non-truthful VCG-based mechanisms. We
generalize these results for affine maximizers.

The second part of this paper proposes a general method for
circumventing the above problem. We introduce a modification of
VCG-based mechanisms in which the agents are given a chance to
improve the output of the underlying algorithm. When the agents
behave truthfully, the welfare obtained by the mechanism is at
least as good as the one obtained by the algorithm's output. We
provide a strong rationale for truth-telling behavior. Our method
satisfies individual rationality as well.
\end{abstract}

\section{Introduction}
Mechanism design is a sub-field of game theory and microeconomics
that studies the design of protocols for non-cooperative
environments. In such environments the participating agents follow
their {\em own} goals and do not necessarily act as instructed by
the mechanism. This theory has traditionally been applied to
economic applications such as auctions of various kinds. An
introduction to mechanism design can be found in several books
\cite{RO94,MWG95}. In recent years, problems on the border of
mechanism design and computer science have attracted the attention
of many researchers, both within and outside the AI community. In
particular, mechanism design models were applied to multi-agent
systems \cite{JG94,Wellmangeb,st97,stijcai2001}, decentralized
resource and task allocations \cite{NR00,Wellmangeb,ess03,PSTR02},
economic and electronic commerce applications \cite{pa99,cra97},
and communication networks \cite{fps00,aks02}.

The canonical mechanism design problem can be described as
follows: A set of rational agents needs to collaboratively choose
an {\em outcome} $o$ from a finite set $O$ of possibilities.  Each
agent $i$ has a privately known {\em valuation function} $v^i:O
\rightarrow R$ quantifying the agent's benefit from each possible
outcome.  The agents are supposed to report their valuation
functions $v^i(\cdot)$ to some centralized mechanism. The goal of
the mechanism is to choose an outcome $o$ that maximizes the {\em
total welfare} $\sum_i v^i(o)$. The main difficulty is that agents
may choose to misreport their valuations in an attempt to affect
the outcome to their liking. Such manipulations are likely to
severely damage the resulting welfare (simulations that
demonstrate this welfare loss can be found in
\shortciteR{grussoCArrol05}). The tool that the mechanism uses to
motivate the agents to reveal the truth is monetary payments.
These payments need to be designed in a way that ensures that
rational agents always reveal their true valuations. Mechanisms
with this property are called {\em incentive compatible} or {\em
truthful} (in dominant strategies). To date, only one general
method, called VCG \cite{Vic61,Cla71,Gro73} (or slightly more
generally, affine maximization), is known for designing such a
payment structure\footnote{Recently, a few truthful mechanisms,
which are not affine maximizers, were obtained for combinatorial
auctions \cite{BGN03}.}. In some settings, it is known that this
method is the sole available one \cite{Rob77,LNM03}.



Many novel applications of mechanism design are complex and
require implementation on computer systems. Cases in point include
combinatorial auctions where multiple items are concurrently sold
in an auction \cite{cabook06}, decentralized task and resource
allocation problems \cite{NR00,Wellmangeb}, and networking
applications \cite{fps00,aks02}.
For many of these applications, the range of possible outcomes is
huge and even finding an outcome that maximizes the total welfare
is known to be NP-complete. Since for such cases computing the
optimal outcome is intractable, the VCG method cannot be applied.

A natural general approach for the development of mechanisms for
such cases would be to use a sub-optimal polynomial time algorithm
for computing the outcome, and calculate the payments by applying
the VCG payment rule to the underlying algorithm.  We term such
mechanisms {\em \vcgbased}.

The starting point of this paper is the observation, noted already
by some researchers \cite{lsl99,NR00}, that \vcgbased\ mechanisms
are not necessarily truthful. Thus, rational agents may lie,
taking advantage of quirks in the outcome determination algorithm.

\subsection{\vcgbased\ Mechanisms are Generally not Truthful}

The first part of this paper examines this last phenomenon in
depth and shows that it is near universal: essentially {\em all}
reasonable \vcgbased\ mechanisms are {\em not} truthful.

We first turn our attention to combinatorial auctions and
characterize the class of truthful \vcgbased\ mechanisms for this
problem\footnote{The importance of combinatorial auctions is
twofold. First, they have direct applications such as FCC
auctions. Second, they abstract many problems of resource
allocation among self-interested agents. A comprehensive survey of
research on combinatorial auctions can be found in a recent book
\cite{cabook06}.}. We say that an allocation algorithm for
combinatorial auctions is {\em reasonable} if, whenever an item is
desired by a single agent only, this agent receives the item. The
above characterization leads to the following corollary:

\smallspace \noindent {\bf Theorem:} Any truthful \vcgbased\
mechanism for combinatorial auctions is not reasonable (unless it
uses the exponential optimal allocation algorithm).

\smallspace In particular, unless $P= NP$, every polynomial time,
truthful \vcgbased\ mechanism is not reasonable. \smallspace

Loosely speaking, we show that essentially the only degree of
freedom available to truthful \vcgbased\ mechanisms is the choice
of {\em range} over which to optimize.  Within this range perfect
optimization is needed.  This theorem seems intuitive as VCG
payments identify each agent's utility with that of society, and
thus if the social welfare is not optimized by the mechanism, then
agents are motivated to to lie in order to do so. Yet, such an
argument only shows that the outcome must be {\em locally} optimal
-- with locality defined as a deviation by a single agent. The
heart of our argument is a delicate hybrid argument showing that
in a general context this local optimization essentially implies
global optimization.

Next we study a family of problems termed {\em cost minimization}
allocation problems.  This family contains many natural
decentralized task allocation problems such as mechanism design
versions of the shortest path problem \cite{ess03,NR00,JG94}. We
call a mechanism for such a problem {\em degenerate} if there
exist inputs that cause it to produce results that are arbitrarily
far from the optimal.

\smallspace \noindent {\bf Theorem:} For any cost minimization
allocation problem, any sub-optimal truthful \vcgbased\ mechanism
is degenerate.

\smallspace A word is in order here about the significance of
these results. \vcgbased\ mechanisms are not just some special
case of truthful mechanisms -- they are essentially the only
general method known for truthful mechanisms in non-single
dimensional settings. Moreover, in certain settings they are known
to indeed be the only truthful mechanisms \cite{Rob77,LNM03}. More
precisely, weighted versions of \vcgbased\ mechanisms -- called
affine maximizers -- are truthful, but our results extend (as we
show) to these cases as well. Consequently, our results imply that
designing truthful mechanisms for computationally intractable
problems requires either restricting the range of the outcomes
(getting ``unreasonable'' or ``degenerate'' mechanisms) or
developing entirely new techniques for truthful mechanisms --
which may not even exist.   A similar implication results if the
intractability stems not from computational considerations, but
rather from communication considerations
\cite[Chapter 11]{cabook06}.

\subsection{The Second Chance Mechanism}

The second part of this paper proposes a general method for
circumventing the difficulty of constructing truthful mechanisms.
While \vcgbased\ mechanisms lose their incentive compatibility,
they still pose a very special property. Loosely speaking, in such
a mechanism, the only reason for an agent to misreport its
valuation is to ``help'' the algorithm compute a better outcome.
We would like to exploit this property to obtain mechanisms that
are ``almost'' truthful.

Given {\em any} algorithm for the corresponding optimization
problem we define the {\em second chance} mechanism based on it.
This mechanism is a modification of the \vcgbased\ mechanism
where, in addition to their valuations, the agents are allowed to
submit {\em appeal functions}. An appeal function allows the agent
to give the algorithm an input (a vector of declared valuations),
which is different from the original input, but without
misreporting its type. When the agents behave truthfully, the
welfare obtained by the mechanism is at least as good as the one
obtained by the algorithm's output.

We then formulate the rationale for truthful behavior in our
mechanism. Informally, our argument is as follows: Under
reasonable assumptions, in any situation in which the agent
believes it is beneficial for it to lie to the mechanism, it is
better for the agent to report its actual type to the mechanism
and ask its appeal to check whether this lie really helps it.
Thus, the agent can construct a truthful strategy premised on the
fact that it is not aware of {\bf any} situation in which another
strategy would be better. We believe that this is a strong
argument for truth-telling.

We construct a version of our mechanism that satisfies individual
rationality as well. A generalization of our results to affine
maximization and to compensation and bonus mechanisms \cite{NR00}
is straightforward.

Several alternative approaches aimed at handling the difficulty of
developing truthful mechanisms were suggested in the past. One
approach is the construction of mechanisms that are
computationally hard to manipulate (e.g.,
\shortciteR{BrtholdiToveyTrick92}). To the best of our knowledge
such manipulations are only hard in the worst case (e.g., it may
be NP-hard to always compute such a manipulation). Nevertheless,
such hardness does not rule out the possibility that manipulations
may be easy to compute in ``typical'' cases. Another possible
approach is to consider other equilibria of VCG \cite{RNDM,RD02}.
However, there is no apparent way for the agents to coordinate
such equilibria. Several recent works construct ascending
mechanisms for combinatorial auctions (e.g. \shortciteR{pa99}).
Such mechanisms rely on assumptions about the agents which are
very different from ours (e.g., myopic behavior). It may be
interesting to compare the virtues of such mechanisms to those of
ours.

A few multi-round mechanisms for combinatorial auctions that let
the agents improve the provisional allocation were proposed and
tested in the past \cite{BLP}. Our argument for truthfulness in
second chance mechanisms may provide a partial explanation for the
relative success reported in these experiments.



\section{Preliminaries}
\label{prelim}

In this section we formally present our model. We attempt as much
as possible to use the standard notions of both mechanism design
and computational complexity theories.

\subsection{Mechanism Design Problems}
\label{Mechanismdesignproblems}

This section formulates the class of mechanism design problems
that we study.

\begin{definition}
\label{MechanismDesignProblem} {\bf (utilitarian mechanism design
problem)} A {\em (utilitarian) mechanism design problem} is
described by:

\begin{enumerate}

\item
A finite set $O$ of {\em allowed outputs}.

\item
Each agent $i = (1, \ldots, n)$ has a real function $v^i(o \in O)$
called its {\em valuation} or {\em type}. This is a quantification
of its benefit from each possible output $o$ in terms of some
common currency. $v^i(.)$ is {\bf privately} known to agent $i$.

\item
If the mechanism's output is $o$ and in addition the mechanism
hands the agent $p^i$ units of currency, then its {\em utility}
$u^i$ equals\footnote{This assumption is called quasi-linearity
and is very common in mechanism design.} $v^i(o) + p^i$. This
utility is what the {\bf agent} aims to optimize.

\item
The goal of the {\bf mechanism} is to select an output $o \in O$
that maximizes the {\em total welfare} $g(v, o) = \sum_i v^i(o)$.

\end{enumerate}
\end{definition}

\noindent An example of such a problem can be found in Section
\ref{example}.

Note that the goal in these problems is  to maximize the total
welfare but not necessarily the revenue. This goal, also known as
economic efficiency, is justified in many settings and is
extensively studied in economics.

In a direct revelation mechanism, the participants are simply
asked to reveal their types to the mechanism. Based on these
declarations the mechanism computes the output $o$ and the payment
$p^i$ for each of the agents.

\begin{definition}
\label{Themechanism} {\bf (mechanism)} A {\em (direct revelation)
mechanism} is a pair $m = (k, p)$ such that:

\begin{itemize}

\item
The {\em output function} $k$ accepts as input a vector $w = (w^1,
\ldots, w^n)$ of declared valuation functions\footnote{We do not
consider the issue of how to represent the valuations.} and
returns an output $k(w) \in O$.

\item
The {\em payment function} $p(w) = (p^1(w), \ldots, p^n(w))$
returns a real vector quantifying the payment handed by the
mechanism to each of the agents (e.g. if $p^i = 2$, the mechanism
pays two units of currency to agent $i$).

\end{itemize}
\end{definition}

The agents try to maximize their own utility and thus may {\bf
lie} to the mechanism. As these lies might severely reduce the
total welfare, the mechanism should be carefully designed such
that it will be for the benefit of the agents to report their
types truthfully.

{\notation} We denote the tuple $(a^1, ... a^{i-1},a^{i+1}, ...,
a^n)$ by $a^{-i}$. We let $(a^i, a^{-i})$ denote the tuple $(a^1,
\ldots, a^n)$.

\begin{definition}{\bf (truthful mechanism)}
A mechanism is called {\em truthful} if truth-telling is a
dominant strategy, i.e., for every agent $i$ of type $v^i$ and for
every type declaration $w^{-i}$ for the other agents, the agent's
utility is {\em maximized} when it declares its real valuation
function $v^i$.
\end{definition}

As an example consider the famous Vickrey auction \cite{Vic61}: A
seller wishes to sell one item in an auction. There are $n$
buyers, each privately knowing its valuation $v^i$ for this item.
(The value for not winning is assumed to be zero.) In a Vickrey
auction each of the buyers is simply asked for its valuation; the
item is allocated to the buyer with the highest bid for the price
of the second highest. The reader may verify that this mechanism
is truthful. Another example of a truthful mechanism can be found
in Section \ref{example}.

In general, the communication protocol of a mechanism can be
complicated. A simple observation known as the {\em revelation
principle} for dominant strategies (e.g., \shortciteR[pp.
871]{MWG95}) states that for every mechanism where the agents have
dominant strategies, there exists an equivalent truthful
mechanism. Thus, w.l.o.g. it is possible to focus on truthful
mechanisms.


\subsection{\vcgbased\ Mechanisms}
\label{VCGbasedmechanisms}

This subsection presents the celebrated VCG mechanisms.
Intuitively, these mechanisms solve utilitarian problems by {\bf
identifying} the utility of truthful agents with the declared
total welfare. We then generalize these mechanisms.

\begin{definition}
\label{VCGmechanism} {\bf (VCG mechanism, (via
\shortciteR{Gro73}))} A mechanism $m = (k , p)$ belongs to the
{\em VCG} family if:
\begin{itemize}

\item
$k(w)$ maximizes the total welfare according to $w$. That is, for
all $w$, $k(w) \in \arg\max_o g(w,o)$.

\item
The payment is calculated according to the {\em VCG formula}:
$p^i(w) = \sum_{j \neq i} w^j(k(w)) + h^i(w^{-i})$ ($h^i(.)$ is an
arbitrary function of $w^{-i}$).

\end{itemize}
\end{definition}

The reader may verify that the Vickrey auction is a VCG mechanism.
It is well known that VCG mechanisms are truthful \cite{Gro73}.

Unfortunately, for many applications, the task of finding an
output $k(w)$ that maximizes the total welfare is computationally
infeasible (e.g., NP-hard). In this paper we consider VCG
mechanisms where the optimal algorithm is replaced by a
sub-optimal but computationally feasible one.

\begin{definition}
\label{VCGBased} {\bf (VCG-based mechanism)} Let $k(w)$ be an
algorithm that maps type declarations into allowable outputs. We
call $m = (k(w), p(w))$ a {\em VCG mechanism based on $k(.)$} if
$p(.)$ is calculated according to the {\em VCG formula}: $p^i(w) =
\sum_{j \neq i} w^j(k(w)) + h^i(w^{-i})$ (where $h^i(.)$ is an
arbitrary function of $w^{-i}$).
\end{definition}

\smallspace Obviously, a VCG-based mechanism that is based on an
optimal algorithm is a VCG mechanism. Note that the payment
function of a VCG-based mechanism is {\bf not} identical to the
VCG payment because the algorithm $k(.)$ is plugged into the
payment formula. We will now characterize the utility of an agent
in VCG-based mechanisms. This utility is equivalent to the total
welfare according to the {\em declared} types of the other agents
and the {\em actual} type of the agent under consideration.

\begin{lemma}\label{LLemmaVCGutil}{\bf (VCG-based utility)} Consider a VCG-based
mechanism defined by the allocation algorithm $k(.)$, and the
functions $(h^1(.), \ldots, h^n(.))$. Suppose that the {\bf
actual} valuation of agent $i$ is $v^i$, and the declarations are
$w = (w^1(.), \ldots, w^n(.))$. Then the {\em utility} of agent
$i$ equals $g((v^i, w^{-i}), k(w)) + h^i(w^{-i})$.
\end{lemma}
\begin{proof} The proof is immediate from the definitions. The
agent's utility equals $v^i(k(w)) + p^i(w) = v^i(k(w)) + \sum_{j
\neq i} v^j(k(w)) + h^i(w^{-i}) = g((v^i, w^{-i}), k(w)) +
h^i(w^{-i})$.
\end{proof}

\smallspace In other words, a VCG-based mechanism {\bf identifies}
the utility of truthful agents with the total welfare. In
particular, when $k(.)$ is optimal, $g((v^i, w^{-i}), k(w))$ is
maximized when the agent is truthful. This implies that VCG
mechanisms are truthful but this truthfulness is not necessarily
preserved by VCG-based mechanisms.

\subsubsection{Example: Non Optimal Vickrey Auction}
\label{propblemsNonOpt}

This section demonstrates the problems that might occur when the
optimal algorithm in a VCG mechanism is replaced by a sub-optimal
one. Consider the sale of a single item. As we already commented,
a Vickrey auction is a VCG mechanism. Its algorithm allocates the
item to the agent with the highest declared value. The function
$h^i(w^{-i}) = -\sum_{j \neq i}w^j(o)$ equals the negation of the
second highest value in case $i$ is winning.

Consider the same mechanism where the optimal algorithm is
replaced by an algorithm that only chooses the second highest
agent. The mechanism will now give the object to the agent with
the second highest declaration for the price of the third highest
agent.

Suppose that there are three agents. Alice who has a value of
$\$2$ million, Bob with a value of $\$1.7$ million, and Carol who
has a value of $\$1$ million. When the agents are truthful, Bob
wins and pays $\$1$ million. In this case it is for Alice's
benefit to reduce her declaration below Bob's. Similarly, if Alice
wins, Bob would like to lower his declaration further, and so on.
Note that there are natural situations where Carol will win as
well.

It is not difficult to see that there are no dominant strategies
in this game. The outcome of the mechanism is highly
unpredictable, depending heavily on the agents' beliefs about the
others, their risk attitude, and their level of sophistication.
Such a mechanism can yield inefficient outcomes. The efficiency
loss may get much worse when the underlying optimization problem
has a complex combinatorial structure (simulations that
demonstrate this in the context of scheduling were done by
\shortciteR{grussoCArrol05}).

\subsubsection{Affine-based mechanisms}
\label{LaffineBasedSec} It is possible to slightly generalize the
class of VCG mechanisms and obtain mechanisms called {\em affine
maximizers}. Such mechanisms maximize affine transformations of
the valuations. When the domain of valuations is unrestricted,
affine maximizers are the sole truthful mechanisms
\cite{Rob77,LNM03}. Similarly to VCG, we generalize these
mechanisms to incorporate sub-optimal algorithms.

\notation\ Let $a = (a_0, \ldots, a_n)$ be an $n+1$-tuple such
that $a_0(.)$ is a valuation function, and $a_1, \ldots, a_n$ are
strictly positive. We define the {\em weighted welfare} $g_a(w,o)$
of an output $o$ as $a_0(o) + \sum_{i > 0} a_i \cdot w_i(o)$ where
$w$ is a vector of types and $o$ an output.

\begin{definition}
\label{LaffineBasedDef} {\bf (affine-based mechanism)} Let $k(w)$
be an algorithm that maps type declarations into allowable
outputs, $a = (a_0, \ldots, a_n)$ be an $n+1$-tuple such that
$a_0(.)$ is a valuation function, and $a_1, \ldots, a_n$ are
strictly positive. We call $m = (k(w), p(w))$ an {\em affine
mechanism based on $k$} if $p$ is calculated according to the
formula: $p^i(w) = \frac{1}{a_i}( \sum_{j \neq i,0} a_j \cdot
w^j(k(w)) + h^i(w^{-i}))$ (where $h^i()$ is an arbitrary function
of $w^{-i}$).
\end{definition}

\smallspace The function $a_0(.)$ can be interpreted as the
preferences of the mechanism over the set of the alternatives and
$a_1, \ldots, a_n$ as weights over the agents. As in VCG
mechanisms, the agents' utility have a convenient
characterization.

\begin{lemma}\label{LAffineUtil}{\bf (affine-based utility)} Consider an affine-based
mechanism defined by the allocation algorithm $k(.)$, a tuple $a$
and the functions $h^1(.), \ldots, h^n(.)$. Suppose that the {\bf
actual} valuation of agent $i$ is $v^i$, and the declarations are
$w = (w^1(.), \ldots, w^n(.))$. Then the {\em utility} of agent
$i$ equals $\frac{1}{a_i}(g_a((v^i, w^{-i}), k(w)) +
h^i(w^{-i}))$.
\end{lemma}
\begin{proof} The proof is immediate from the definitions. The
agent's utility equals $v^i(k(w) + p^i(w)) = \frac{1}{a_i}(a_i
v^i(k(w)) + p^i(w) = \frac{1}{a_i}(g_a((v^i, w^{-i}), k(w)) +
h^i(w^{-i}))$.
\end{proof}

\smallspace In other words, an affine-based mechanism identifies
the agents' utility with the affine transformation of the
valuations it aims to optimize. In particular, when $k(.)$
maximizes $g_a(w,.)$, the mechanism is truthful.

\subsection{Computational Considerations in Mechanism Design}
\label{Computationalconsiderations}
This section adopts standard notions of computational complexity
to revelation mechanisms.

\begin{definition}{\bf (polynomial mechanism)}
A mechanism $(k,p)$ is called {\em polynomial time computable} if
both $k(w)$ and $p(w)$ run in polynomial time (in the size of
$w$).
\end{definition}

Note that a \vcgbased\ mechanism is polynomial iff its output
algorithm and the functions $h^i(.)$ are polynomial. We sometimes
call polynomial algorithms and mechanisms computationally
feasible.

\begin{definition}{\bf (NP-complete problem)}
A mechanism design problem is called {\em NP-Complete} if the
problem of finding an output that maximizes the total welfare is
NP-Complete.
\end{definition}

We use the term {\em feasible} to denote ``acceptable''
computational time and {\em infeasible} for unacceptable
computational time. In particular, NP-hard problems and
exponential algorithms are considers infeasible, while polynomial
algorithms are considered feasible. We use these non-standard
terms because most of our results are not limited to specific
complexity classes.

\subsection{Example: Combinatorial Auctions}
\label{example}

The problem of combinatorial auctions has been extensively studied
in recent years (a recent book can be found at
\shortciteR{cabook06}). The importance of this problem is twofold.
Firstly, several important applications rely on it (e.g., the FCC
auction \shortciteR{cra97}). Secondly, it is a generalization of
many other problems of interest, in particular in the field of
electronic commerce.

\smallspace {\bf The problem:} A seller wishes to sell a set $S$
of items (radio spectra licenses, electronic devices, etc.) to a
group of agents who desire them. Each agent $i$ has, for every
subset $s \subseteq S$ of the items, a non-negative number
$v^i(s)$ that represents how much $s$ is worth for it. $v^i(.)$ is
privately known to each agent $i$. We make two standard additional
assumptions on the agents' type space:

\begin{description}

\item [No externalities] The valuation of each agent $i$ depends
only on the items allocated to it. In other words, for every two
allocations $x = (x_1, \ldots, x_n)$ and $y = (y_1, \ldots, y_n)$,
if $x_i = y_i$, then $v^i(x) = v^i(y)$. Thus, we denote the
valuation of each agent $i$ by $v^i:2^S \rightarrow R$.

\item [Free disposal] Items have non-negative values, i.e., if $s
\subseteq t$, then $v^i(s) \leq v^i(t)$. Also, $v^i(\phi) = 0$.
\end{description}

\smallspace Items can either be complementary, i.e., $v^i(S
\bigcup T) \geq v^i(S) + v^i(T)$, or substitutes, i.e., $v^i(S
\bigcup T) \leq v^i(S) + v^i(T)$ (for disjoint $S$ and $T$). For
example, a buyer may be willing to pay $\$200$ for a T.V. set,
$\$150$ for a VCR, $\$450$ for both and only $\$200$ for two VCRs.

If agent $i$ gets the set $s^i$ of items, and its payment is
$p^i$, its utility is  $v^i(s^i) + p^i$. (The payments in
combinatorial auctions are non-positive.) This utility is what
each agent tries to optimize. For example, an agent prefers to buy
a $\$1000$ valued VCR for $\$600$, gaining $\$400$, rather buy a
$\$1500$ valued VCR for $\$1250$.

In a VCG mechanism for a combinatorial auction, the participants
are first required to reveal their valuation functions to the
mechanism. The mechanism then computes, according to the agents'
declarations, an allocation $s$ that maximizes the total welfare.
The payment for each of the agents is calculated according to the
VCG formula. By Lemma \ref{LLemmaVCGutil}, the utility $u^i =
v^i(s^i) + p^i$ of each of the agents is maximized when it reveals
its true valuation to the mechanism. When all agents are truthful,
the mechanism maximizes the total welfare.

Consider, however, the computational task faced by such a
mechanism. After the types are declared, the mechanism needs to
select, among all possible allocations, one that maximizes the
total welfare. This problem is known to be NP-Complete. Therefore,
unless the number of agents and items is small, such a mechanism
is computationally infeasible. Even the problem of finding an
allocation that approximates the optimal allocation within a
reasonable factor of $\sqrt{|S|} - \epsilon$ is $NP$-Complete
\cite{zuckerman06}. Nevertheless, various heuristics and tractable
sub-cases have been analyzed in the literature \cite[Chapter
13]{cabook06}. We would like to find a way to turn these
sub-optimal algorithms into mechanisms.

We note that, in general, revealing a valuation function requires
exponential communication. While we ignore communication issues in
this paper, a subsequent work \cite{R01b} extends the second
chance method to address communication limitations as well.


\section{Limitations of Truthful \vcgbased\ Mechanisms}
\label{limitations}

This section studies the limitations of truthful \vcgbased\
mechanisms. Section \ref{tvbcomb} characterizes these mechanisms
for the important problem of combinatorial auctions (see Section
\ref{example}). This characterization precludes the possibility of
obtaining truthfulness by applying VCG rules to many of the
proposed heuristics for combinatorial auctions (e.g., the greedy
algorithms of \shortciteR{lsl99} and \shortciteR{NN00}). Moreover,
we show that any truthful non-optimal \vcgbased\ mechanism for
combinatorial auctions suffers from abnormal behavior. Section
\ref{campsec} shows that for many natural cost minimization
problems, any truthful \vcgbased\ mechanism is either optimal or
produces results that are arbitrarily far from optimal. As a
result, when such a problem is computationally intractable, any
truthful computationally feasible \vcgbased\ mechanism has inputs
that cause it to produce degenerate results. Furthermore, since
standard algorithmic techniques do not yield such anomalies, it
might be difficult to develop algorithms that can be plugged into
truthful mechanisms. We generalize these results to affine-based
mechanisms as well.

\subsection{Truthful \vcgbased\ Mechanisms for Combinatorial Auctions}
\label{tvbcomb}

This section characterizes the class of truthful \vcgbased\
mechanisms for combinatorial auctions.

\begin{definition}
\label{maxalg} {(\bf maximal in its range)} Let $k(w)$ be an
algorithm that maps type declarations into allowable outputs. Let
$V { \stackrel { df } { = } } \prod_{i = 1}^n V^i$ be the space of
all possible types and let $V' \subseteq V$ be a subspace of $V$.
Let ${\cal{O}}$ denote the range of $k$ at $V'$, i.e. ${\cal{O}} =
\{ k(w) | w \in V' \} $. We say that $k$ is {\em \maximizer\ at
$V'$} if for every type $w \in V'$, $k(w)$ maximizes $g$ over
${\cal{O}}$. We say that $k$ is {\em \maximizer\ } if it is
\maximizer\ at $V$.
\end{definition}

Consider, for example, an algorithm for combinatorial auctions
that allocates all the items (the set $S$) to the agent with the
highest valuation $v^i(S)$. Clearly, this polynomial time
algorithm is \maximizer\ . The welfare obtained by the allocation
of this algorithm achieves at least a factor of $\max(1/n, 1/|S|)$
of the optimal welfare (where $n$ denotes the number of agents).

\begin{proposition}\label{maxRangeTruthful}
A \vcgbased\ mechanism with an output algorithm that is
\maximizer\ is truthful.
\end{proposition}

\begin{proof}
Such a mechanism is a VCG mechanism where the set of allowable
outputs is the range of its output algorithm. By Lemma
\ref{LLemmaVCGutil} such a mechanism is truthful.
\end{proof}

\smallspace We will now show that the above proposition almost
characterizes the class of truthful \vcgbased\ mechanisms for the
combinatorial auction problem.

\notation\ We let $\tilde{V}$ denote the space of all types $v =
(v^1, \ldots, v^n)$ such that for any two different allocations
$x$ and $y$, $g(v,x) \neq g(v,y)$. (Recall that $g(.)$ denotes the
total welfare.)

\smallspace It is not difficult to see that $\tilde{V}$ contains
almost all the types, i.e. $V - \tilde{V}$ has a measure zero in
$V$.

\begin{theorem}
\label{CharacterizationComb} If a  \vcgbased\ mechanism for the
combinatorial auction problem is truthful, then its output
algorithm is \maximizer\ at $\tilde{V}$.
\end{theorem}

\begin{proof}
Assume by contradiction that  $m = (k, p)$ is truthful but $k(.)$
is not \maximizer\ at $\tilde{V}$. Since the functions $h^i(.)$ do
not affect the truthfulness of the mechanism, we can assume that
they are all zero, i.e., we assume that for all $i$, $p^i(w) =
\sum_{j \neq i} w^j(k(w))$. According to Lemma
\ref{LLemmaVCGutil}, the utility of each agent $i$ equals
$v^i(k(w)) + \sum_{j \neq i} w^j(k(w)) = g((v^i, w^{-i}), k(w))$.

Let ${\cal{O}}$ denote the range of $k(.)$ at $\tilde{V}$ and let
$v \in \tilde{V}$ be a type such that $k(v)$ is not optimal over
${\cal{O}}$. Let  $y = \arg\max_{o \in \cal{O}} g(v,o)$ be the
optimal allocation among ${\cal{O}}$. Note that from the
definition of $\tilde{V}$, $y$ is unique. Finally, let $w \in
\tilde{V}$ be a type such that $y = k(w)$. Such a type exists
since $y$ is in the range of the algorithm.

Define a type vector $z$ by

\[ z^i(s) = \left\{ \begin{array}{ll}
  v^i(s)  & \mbox{if $s \not\supseteq y^i$} \\
  \alpha        & \mbox{if $s \supseteq y^i$.}
\end{array}
\right. \]

\noindent where $\alpha$ stands for a sufficiently large number.
In other words, each agent $i$ strongly desires the set $y^i$.
Apart from this, $v^i$ and $z^i$ are identical. We assume that $z
\in \tilde{V}$. Otherwise we could add sufficiently small
``noise'' $\epsilon^i(s)$ to $z$ such that all the claims below
remain true.

We will show that $z$ ``forces'' the algorithm to output $y$. We
will then show that if the algorithm outputs $y$ when the type is
$z$, it must also output $y$ when the type is $v$ -- a
contradiction.

\begin{lemma}
\label{charlemma} $y = k(z)$.
\end{lemma}

\begin{proof}
Define a sequence of type vectors by:

\[ \begin{array}{ll}
  w_0 =   & (w^1, \ldots, w^n)  \\
  w_1 =      & (z^1, w^2, \ldots, w^n)  \\
  w_2 =      & (z^1, z^2,w^3, \ldots, w^n)  \\
\vdots \\
  w_n =  & (z^1, \ldots, z^n).  \\
\end{array}
\]

\noindent In other words, every agent in turn moves from $w^i$ to
$z^i$ . We assume that $w_j \in \tilde{V}$ for all $j$. It is not
difficult to see that $z$ can be modified by adding small noise to
it, in a way that guarantees the above.

\begin{claim}
$k(w_1) = y$.
\end{claim}

\begin{proof}
Assume by contradiction that this is false. From the definition of
$\tilde{V}$ we obtain that $g(w_1, k(w_1)) \neq g(w_1, y)$.

Consider the case where agent $1$'s type is $z^1$ and the types of
the others are $w^2, \ldots, w^n$. By declaring $w^1$, agent $1$
can force the algorithm to decide on $y$. Since the mechanism is
truthful, it must be that $g(w_1, k(w_1)) > g(w_1, y)$.

Since $\alpha$ is large, it must be that $k^1(w_1) \supseteq y^1$
(i.e., agent 1 gets all the items it gets when its type is $w^1$).
Thus, from the definition of $z^1$, we obtain $\alpha + \sum_{j =
2}^n w^j(k(w_1)) > \alpha + \sum_{j = 2}^n w^j(y)$. Because, due
to the free disposal assumption, $w^1(k(w_1)) \geq w^1(y)$, we
obtain that $w^1(k(w_1)) + \sum_{j = 2}^n w^j(k(w_1))
> w^1(y) + \sum_{j = 2}^n w^j(y)$ (even when $z$ is perturbed). Thus, $g(w_0, k(w_1))
> g(w_0, y)$.

Therefore, when the type of agent $1$ is $w^1$, it is better off
declaring $z^1$, forcing the mechanism to output $k(w_1)$. This
contradicts the truthfulness of the mechanism.
\end{proof}

\smallspace Similarly, by induction on $j$, we obtain that $k(w_j)
= y$ for all $j$, and in particular for $w_n = z$. This completes
the proof of Lemma \ref{charlemma}.
\end{proof} 

\smallspace We will now show that $k(z) = y$ implies that $k(v) =
y$ -- a contradiction. Consider the following sequence of type
vectors:

\[ \begin{array}{ll}
  v_0 =   & (v^1, \ldots, v^n)  \\
  v_1 =      & (z^1, v^2, \ldots, v^n)  \\
\vdots \\
  v_n =  & (z^1, \ldots, z^n).  \\
\end{array}
\]

\noindent In other words, every agent in turn, moves from $v^i$ to
$z^i$. Again we can choose $z$ such that all $v_j$s are in
$\tilde{V}$.

\begin{claim}
For all $v_j$, $y$ maximizes $g$ on ${\cal{O}}$.
\end{claim}

\begin{proof}
We will show this for $v_1$. The proof for $j > 1$ follows from a
similar argument. Assume by contradiction that $x \neq y$
maximizes the welfare for $v_1$. Since $\alpha$ is arbitrarily
large it must be that $x^1 \supseteq y^i$. Consequently, in both
cases agent 1's valuation equals $\alpha$.

Recall that $y$ uniquely maximizes $g$ on ${\cal{O}}$ for $v_0$.
Thus, for every allocation $x \neq y$, we have $v^1(y) + \sum_{j =
2}^n v^j(y) > v^1(x) + \sum_{j = 2}^n v^j(x)$. Therefore, $\alpha
+ \sum_{j = 2}^n v^j(y) > \alpha + \sum_{j = 2}^n v^j(x)$. But the
left hand side equals $g(v_1,y)$ and the right hand side equals
$g(v_1,x)$. Thus, $g(v_1,y) > g(v_1,x)$ -- contradiction.
\end{proof}

\begin{claim}
$k(v_{n-1}) = y$.
\end{claim}

\begin{proof}
We showed that $k(v_n) = y$. (Recall that $v_n = z$.) We also
showed that $y$ uniquely maximizes $g(v_{n-1},.)$. Let $x_{n-1} =
k(v_{n-1})$. Assume by contradiction that $x_{n-1} \neq y$.
According to Lemma \ref{LLemmaVCGutil}, the utility of agent $n$
when it is truthful is $g(v_{n-1}, x_{n-1})$. Thus, when agent
$n$'s type is $v^n$, it is better off declaring $z^n$ obtaining a
utility of $g(v_{n-1}, y)$. This contradicts the truthfulness of
the mechanism.
\end{proof}

\smallspace Similarly, by downward induction on $j$, we obtain
that $k(v_0) = y$. But $v_0 = v$ and we assumed that $k(v) \neq y$
-- a contradiction. This completes the proof of Theorem
\ref{CharacterizationComb}.
\end{proof} 

\smallspace {\bf Remarks} The above theorem characterizes the
output algorithms that could be incorporated into truthful
\vcgbased\ mechanisms on all but a zero-measured subset of the
types. This characterization holds even when the set of possible
types is discrete (under the mild condition that the type vector
$z$ can be defined such that the agents are not indifferent
between allocations). The theorem gives rise to several
interesting algorithmic and combinatorial questions. For example,
given an approximation factor $c \leq 1$, what is the minimal size
of a sub-family ${\cal{O}} \subseteq O$ such that for every $v$,
$\max_{y \in {\cal{O}}} g(v, y) \geq c \cdot g_{opt}(v)$? A
limited version of this question was analyzed by \shortciteR{RNDM}
and \shortciteR{RD02}.

\begin{corollary}
\label{charcomp} Consider a \vcgbased\ mechanism for a
combinatorial auction with an output algorithm $k$. If the
mechanism is truthful, there exists an output algorithm
$\tilde{k}$, \maximizer , such that for every $v$, $g(v, k(v)) =
g(v, \tilde{k}(v))$.
\end{corollary}

\begin{proof}
Let $\cal{O}$ denote the range of $k$ on $\tilde{V}$, and define
another algorithm that is optimal in its range by $\forall v,
\tilde{k}(v) \in \arg\max_{o \in {\cal{O}}}$. According to
Proposition \ref{maxRangeTruthful}, a VCG mechanism based on
$\tilde{k}$ is truthful. Consider the case where the agents are
truthful. Recall that the utility of all agents is determined by
the resulting total welfare. Thus, it is not difficult to see that
the welfares $g(v, k(v))$ and $g(v, \tilde{k}(v))$ must be
continuous in $v$. Two continuous real functions, which are
identical on a dense subspace, are identical on the whole space
and thus the corollary follows.
\end{proof}

\smallspace We now show that non-optimal truthful \vcgbased\
mechanisms suffer from the following disturbing abnormal behavior:

\begin{definition}\label{LReasonableMechCA}
\label{resmech} {\bf (reasonable mechanism)} A mechanism for
combinatorial auctions is called {\em reasonable} if whenever
there exists an item $j$ and an agent $i$ such that:
\begin{enumerate}
\item For all $S$, if $j \notin S$, then $v^i(S \cup \{ j \}) >
v^i(S)$, and,

\item For every agent $l \neq i$, $\forall S, \, v^i(S \cup \{ j
\}) = v^i(S)$,
\end{enumerate}

then $j$ is allocated to agent $i$.
\end{definition}

\noindent Simply put, in situations where only one agent desires
an item, that agent gets it.

\begin{theorem}
\label{nonres} Any non-optimal truthful \vcgbased\ mechanism for
combinatorial auctions is not reasonable.
\end{theorem}

\begin{proof}
Consider such a mechanism $m$. According to Corollary
\ref{charcomp} there exists an equivalent mechanism $\tilde{m} =
(\tilde{k}, p)$, which is optimal in its range. Since $\tilde{m}$
must also be sub-optimal, there exists at least one partition $s =
(s^1, \ldots, s^n)$ that is not in the range of the mechanism.
Define a vector of types by:
\[ v^i(x) = \left\{ \begin{array}{ll}
  1  & \mbox{if $x \supseteq s^i$} \\
  0  & \mbox{otherwise.}
\end{array}
\right. \]
In other words, each agent $i$ wants a single set $s^i$, and no
two agents want the same item (as the sets are disjoined). Since
$s$ is not in the range of $\tilde{m}$ it must be that
$\tilde{k}(v) \neq s$. Since $s$ is strictly optimal, $k(v)$ must
also be suboptimal. Hence, there exists at least one agent $i$
that does not get $s^i$. In particular, there exists at least one
item $j \in s^i$ that agent $i$ does not get. Since $i$ is the
only agent that desires $j$, the theorem follows.
\end{proof}

\begin{corollary}
Unless $P = NP$,  any polynomial time truthful \vcgbased\
mechanism for combinatorial auctions is not reasonable.
\end{corollary}

\smallspace We believe that most of the ``natural'' allocation
algorithms (e.g., linear programming relaxations, algorithms that
greedily allocate items to agents,  local search algorithms) will
not yield the above anomaly. In particular, we presume that when
each agent wants a single subset of items and these subsets are
disjoined, any such algorithm will find the optimal allocation.
Thus, the above corollary suggests that it might be difficult to
develop allocation algorithms that yield truthful \vcgbased\
mechanisms.

We now show how to generalize our results to any affine-based
mechanism. Given a tuple $a = (a_0, \dots, a_{n})$, we define
$\tilde{V}$ to be the space of all types $v$ such that for any two
different allocations $x$ and $y$, $g_a(v,x) \neq g_a(v,y)$.
Similarly to the unweighted case, we say that an algorithm is
optimal in its range with respect to $g_a(.)$ if it always
produces allocations that maximize $g_a(.)$.

\begin{theorem}
\label{CharacterizationAffine} Consider an affine-based mechanism
for the combinatorial auction problem defined by an allocation
algorithm $k(.)$, and a tuple $a = (a_0, \dots, a_{n})$. If the
mechanism is truthful, then $k(.)$ is maximal in its range with
respect to $g_a(.)$ at $\tilde{V}$.
\end{theorem}

\begin{proof} {\bf (sketch)} The proof is similar to the proof of
Theorem \ref{CharacterizationComb} and we thus only sketch it.
Define $\tilde{V}$ and ${\cal{O}}$ similarly but w.r.t. the affine
transformation $g_a(.)$. Assume by contradiction that there exists
a type vector $v$ such that $k(v)$ is not optimal in ${\cal{O}}$.
Let $y$ be the optimal allocation in the range ${\cal{O}}$, and $w
\in \tilde{V} $ such that $k(w) = y$. According to Lemma
\ref{LAffineUtil}, the utility of each agent is maximized with the
weighted welfare $g_a((v^i, w^{-i}),.)$. Thus, it is possible to
proceed along the lines of the proof of Theorem
\ref{CharacterizationComb}: Define a type vector $z$ similarly;
then, start from $w$ and gradually transform all agents to $z$ and
conclude that $k(z) = y$; then gradually transform all agents from
$z$ to $v$ and show that $k(v) = y$, i.e., -- a contradiction.
\end{proof}

\smallspace {\bf Open Questions} We currently do not know whether
theorems similar to Theorem \ref{CharacterizationComb} hold when
the valuations are bounded. Moreover, we do not know how to get
rid of the usage of $\tilde{V}$. Thus, we do not preclude the
possibility that Corollary \ref{charcomp} will not hold when the
space of possible types is discrete. We also do not know whether
they hold when the allocation algorithm is randomized or whether
Bayesian versions of our theorems apply to the expected
externality mechanism \cite{GerardVaret} (an analog of VCG in the
Bayesian model). We leave this to future research. We conjecture
that similar theorems apply to many other mechanism design
problems.

\subsection{Truthful \vcgbased\ Mechanisms for Cost Minimization
Problems} \label{campsec}

We now show that for many natural cost minimization problems, any
truthful \vcgbased\ mechanism is either optimal or produces
results that are arbitrarily far from the optimal. We start with a
sample problem.

\smallspace {\bf Multicast transmissions:} A communication network
is modeled by a directed graph $G = (V, E)$. Each edge $e$ is a
privately owned link. The cost $t_e$ of sending a message along
that edge is {\em privately} known to its owner. Given a source $s
\in V$ and a set $T \subseteq V$ of terminals, the mechanism must
select a subtree rooted in $s$ that covers all the terminals. The
message is then broadcasted along this tree. We assume that no
agent owns a cut in the network.

Naturally, the goal of the mechanism is to select, among all
possible trees, a tree $R$ that minimizes the total cost: $\sum_{e
\in R} t_e$. The goal of each agent is to maximize its {\em own}
profit: $p^i - \sum_{(e \in R {\mbox{ owned by\ }} i)} t_e$. It is
not difficult to see that this is a utilitarian mechanism design
problem.

\smallspace This example was introduced by \shortciteR{fps00}
(using a different model). It is motivated by the need to
broadcast long messages (e.g., movies) over the Internet. We now
generalize this example.

\begin{definition}
\label{CAMP} {\bf (cost minimization allocation problem)}

A {\em cost minimization allocation problem (CMAP)} is a mechanism
design problem described by:

\begin{description}
\item[Type space] The type of each agent $i$ is described by a
vector $(v^i_1, \ldots, v^i_{m_i})$.  We let $m = \sum_i m_i$. (In
our multicast example $v^i_e$ corresponds to the negation of the
cost $t_e$.)

\item[Allowable outputs] Each output is denoted by a bit vector $x
= (x^1_1, \ldots, x^1_{m_1}, \ldots, x^n_1, \ldots, x^n_{m_n}) \in
\{0,1\}^m$. We denote $(x^i_1, \ldots, x^i_{m_i})$ by $x^i$. There
may be additional constraints on the set $O$ of allowable outputs.
(In our example $x$ corresponds to a tree in the network's graph
where $x^i_j$ equals $1$ iff the corresponding edge is in the
chosen tree.)
\end{description}
such that the following conditions are satisfied:

\begin{description}
\item[Unbounded costs]
If $v^i = (v^i_1, \ldots, v^i_{m^i})$ describes a type for agent
$i$ and $w^i \leq v^i$ (as vectors), then $w^i$ also describes a
type.

\item[Independence and monotonicity] Each valuation $v^i$ depends
only on $i$'s bits $x^i$. (In our example, the agent valuation of
a given tree depends only on its own edges in it.) If for all $j$,
$w^i_j \leq v^i_j$, then for every output $x$, $w^i(x^i) \leq
v^i(x^i)$.

\item[Forcing condition]
For every type $v$,  an allowable output $x$ and a real number
$\alpha$,

define a type $v [ \alpha ] $ by

\[ v[\alpha]^i_j = \left\{ \begin{array}{ll}
  v^i_j  & \mbox{if $x^i_j = 1$} \\
  \alpha        & \mbox{otherwise.}
\end{array}
\right. \]

The forcing condition is satisfied if for every allowable output
$y \neq x$, $\lim_{\alpha \rightarrow -\infty} g(t(\alpha), y) =
-\infty$.

\end{description}
\end{definition}

Many natural decentralized task allocation problems in which the
goal is to minimize the total cost under given constraints belong
to this class. In particular the reader may verify that our
multicast example falls into this category. Another example is the
shortest path problem studied extensively in recent years (e.g.,
\shortciteR{JG94,AT02,ess03}).

\notations\ For a type $v$ we let $g_{opt}(v)$ denote the optimal
value of $g$. We denote $g(v, k(v))$ by $g_k(v)$.

\begin{definition}
\label{degenerate} {\bf (degenerate algorithm)} An output
algorithm $k$ is called {\em degenerate} if the ratio  $r_k(v) =
\frac {g_k(v) - g_{opt}(v)} {|g_{opt}(v)| + 1}$ is unbounded,
i.e., there exist $v$'s such that $r_k(v)$ is arbitrarily large.
\end{definition}

A degenerate algorithm is arbitrarily far from optimal, both
additively and multiplicatively. Note that this should not be
confused with the standard notion of an approximation ratio, as
our definition corresponds to a single problem. In particular, the
number of agents is fixed. We note that we do not rule out the
possibility that such an algorithm will be good by some non worst
case metric.

\begin{theorem}
\label{degopt} If a \vcgbased\ mechanism for a CMAP is truthful,
then its output algorithm is either optimal or degenerate.
\end{theorem}

Before stating the proof let us illustrate it using the multicast
transmission example. Suppose that we start with a type vector
that leads to a sub-optimal solution. If we raise the cost of an
edge, the utility of the owner cannot increase (due to the
truthfulness and Lemma \ref{LLemmaVCGutil}). We then gradually
raise the cost of all edges except the ones in the optimal tree.
Still, the algorithm will have to choose a sub-optimal tree.
However, the cost of {\em any} suboptimal tree is now arbitrarily
high while the optimal cost remains the same. \smallspace

\begin{proof}
Let $m = (k, p)$ be a  non-optimal truthful \vcgbased\ mechanism
for a CMAP. As in Theorem \ref{CharacterizationComb}, assume that
$p^i(w) =  \sum_{j \neq i} w^j(k(w))$. Let $v$ be a type vector
such that $k(v)$ is not optimal and let $y = opt(v)$ be an optimal
output.

We define a type $z$ by:

\[ z^i_j = \left\{ \begin{array}{ll}
  v^i_j  & \mbox{if $y^i_j = 1$} \\
  -\alpha       & \mbox{otherwise.}
\end{array}
\right. \]

\noindent where $\alpha$ is arbitrarily large.

Consider the type sequence:

\[ \begin{array}{ll}
  v_0 =   & (v^1, \ldots, v^n)  \\
  v_1 =      & (z^1, v^2, \ldots, v^n)  \\
\vdots \\
  v_n =  & (z^1, \ldots, z^n).  \\
\end{array}
\]

\begin{claim}
For all $j$, $y = opt(v_j)$.
\end{claim}

\begin{proof}
By definition $y$ is optimal for $v_0$. Let $x \neq y$ be an
allocation. From the independence condition, for all $j$, $g(v_j,
y) = g(v_0, y)$. From the monotonicity, $g(v_j, x) \leq g(v_0,
x)$. Together, $g(v_j, x) \leq g(v_0, x) \leq g(v_0, y) = g(v_j,
y)$.
\end{proof}

\begin{claim}
$g(v_1, k(v_1)) < g(v_1, y)$
\end{claim}

\begin{proof}
Assume by contradiction that the claim is false. Since $y$ is
optimal for $v_1$, this means that $g(v_1, k(v_1)) = g(v_1, y)$.
From independence, $g(v_1, y) = g(v_0, y)$. Recall that $k(v_0)$
is suboptimal so $g(v_0, y) > g(v_0, k(v_0))$. From monotonicity
(we only worsen the type of agent $1)$, $g(v_0, k(v_1)) \geq
g(v_1, k(v_1))$. Thus, together $g(v_0, k(v_1))  \geq g(v_1,
k(v_1)) = g(v_1, y) = g(v_0, y)
> g(v_0, k(v_0))$. In particular, $g(v_0, k(v_1)) > g(v_0, k(v_0))$.

Consider the case where agent $1$'s type is $v^1$ and the
declarations of the other agents are $(v^2, \ldots, v^n)$.
According to Lemma \ref{LLemmaVCGutil}, its utility when it is
truthful, equals $g(v_0, k(v_0))$. On the other hand, when it
falsely declares $z^1$, its utility equals $g(v_0, k(v_1))$. Since
we showed that $g(v_1, k(v_1)) > g(v_0, k(v_0))$, this contradicts
the truthfulness of the mechanism.
\end{proof}

\smallspace Similarly, we obtain that $g(v_n, k(v_n)) < g(v_n, y)
= g(v_0, y)$. By the forcing condition, $g(v_n, k(v_n))
\rightarrow -\infty$ when $\alpha \rightarrow \infty$. Thus, the
algorithm is degenerate.
\end{proof}

\begin{corollary}
Unless $P = NP$, any polynomial time truthful \vcgbased\ mechanism
for an NP-hard CAMP is degenerate.
\end{corollary}\qed

\smallspace Note that due to the revelation principle, the
theorems in this section hold for any mechanism where the agents
have dominant strategies. Similarly to Theorem \ref{degopt}, any
mechanism that uses VCG payments and has a non-optimal ex-post
Nash equilibrium  also has equilibria which are arbitrarily far
from optimal.

We now show how to generalize the theorems in this section to
affine-based mechanisms.

\begin{theorem}
\label{degoptaffine} If an affine-based mechanism $(k,p)$ for a
CMAP is truthful, then its output algorithm is either optimal or
degenerate.
\end{theorem}
\begin{proof}{\bf(sketch)} The proof is almost identical to the
proof of Theorem \ref{degopt}. Let $v$ be a type such that $k(w)$
is not optimal w.r.t. to the corresponding affine transformation
$g_a$. We define a type vector $z$ similarly to Theorem
\ref{degopt} and consider a sequence of type vectors where each
agent in turn changes its type from $w^i$ to $z^i$. Due to the
incentive compatibility and Lemma \ref{LAffineUtil} , the utility
of each agent cannot increase, meaning that the weighted welfare
$g_a$ remains sub-optimal. Due to the forcing condition, all
outputs except the optimal, have arbitrarily high cost. This means
that the algorithm is degenerate.
\end{proof}

\cout{\smallspace {\bf Remarks} A randomized algorithm is a
distribution over deterministic ones. It is possible to define an
affine-based mechanism similarly when the allocation algorithm is
randomized. Call a randomized algorithm degenerate if there exist
inputs that cause the ratio between its expected weighted welfare
and the optimum to be arbitrarily high. If the algorithm is not
optimal, the only way it can ``escape'' degeneracy is by assigning
infinitely small probabilities to the ``forbidden'' outputs. Thus,
any truthful suboptimal VCG-based mechanism which has a
computational time (or number of coin tosses) which is not
dependent on the {\em size} of the types (e.g. the number of bits
of the actual costs in multicast transmission) is degenerate.}

\smallspace The compensation and bonus mechanism \cite{NR00}
identifies the utility of agents with the total welfare similarly
to VCG, i.e., the utility of an agent can be described similarly
to Lemma \ref{LLemmaVCGutil}. Thus, all the theorems in this
section can be applied to compensation and bonus mechanisms as
well.


\section{Second Chance Mechanisms}
\label{compfeas}

To date, affine maximization is the only known general method for
the development of truthful mechanisms. Therefore, the results in
the previous section do not leave much hope for the development of
truthful mechanisms for many complex problems.

This section proposes a method for circumventing this problem.
Consider a \vcgbased\ mechanism. An immediate consequence of Lemma
\ref{LLemmaVCGutil} is that the only reason for an agent to
misreport its type is to help the algorithm to improve the overall
result. This leads to the intuition that if the agents {\bf
cannot} improve upon the underlying algorithm, they can do no
better than be truthful. We would like to exploit this special
property of \vcgbased\ mechanisms and construct mechanisms that
are ``almost'' truthful.

Given {\em any} algorithm for the corresponding optimization
problem we define the {\em second chance} mechanism based on it.
This mechanism is a modification of the \vcgbased\ mechanism where
in addition to their valuations, the agents are allowed to submit
{\em appeal functions}. An appeal function allows the agent to
give the algorithm an input (vector of declared valuations) that
is different from the original input but without misreporting its
type. When the agents behave truthfully, the welfare obtained by
the mechanism is at least as good as the one obtained by the
algorithm's output.

We then formulate the rationale for truthfulness in second chance
mechanisms. Informally, our argument is as follows: Under
reasonable assumptions, in any situation in which the agent
believes it is beneficial for it to lie to the mechanism, it is
better for it to report its actual type to the mechanism and ask
its appeal to {\em check} whether this lie is indeed helpful.
Thus, the agent can construct a truthful strategy premised on the
fact that it is not aware of {\bf any} situation in which another
strategy is better for it. We believe that this is a strong
argument for truth-telling.

A generalization of our results to affine maximization and to
compensation and bonus mechanisms is straightforward.

\subsection{The Mechanism}
In this section we formulate the second chance mechanism and its
basic properties.

\begin{definition}
\label{appealDef} {\bf (appeal function)} Let $V = \prod_i V_i$
denote the type space of the agents. An {\em appeal} is a partial
function\footnote{A function $f:D \rightarrow R$ is called {\em
partial} if its domain is a subset of $D$, i.e. if $Dom(f)
\subseteq D$.} $l:V \rightarrow V$.
\end{definition}

The semantics of an appeal $l(.)$ is: ``when the agents' type
vector is $v = (v_1, \ldots, v_n)$, I believe that the output
algorithm $k(.)$ produces a better result (w.r.t. $v$) if it is
given the input $l(v)$ instead of the actual input $v$''. An
appeal function gives the agent an opportunity to improve the
algorithm's output. If $v$ is not in the domain of $l(.)$, the
semantics is that the agent does not know how to cause the
algorithm to compute a better result than $k(v)$.

The second chance mechanism is defined in Figure
\ref{fig:secChance}. It is a modification of VCG that allows the
agents to submit appeal functions as well.

\begin{figure} [h]
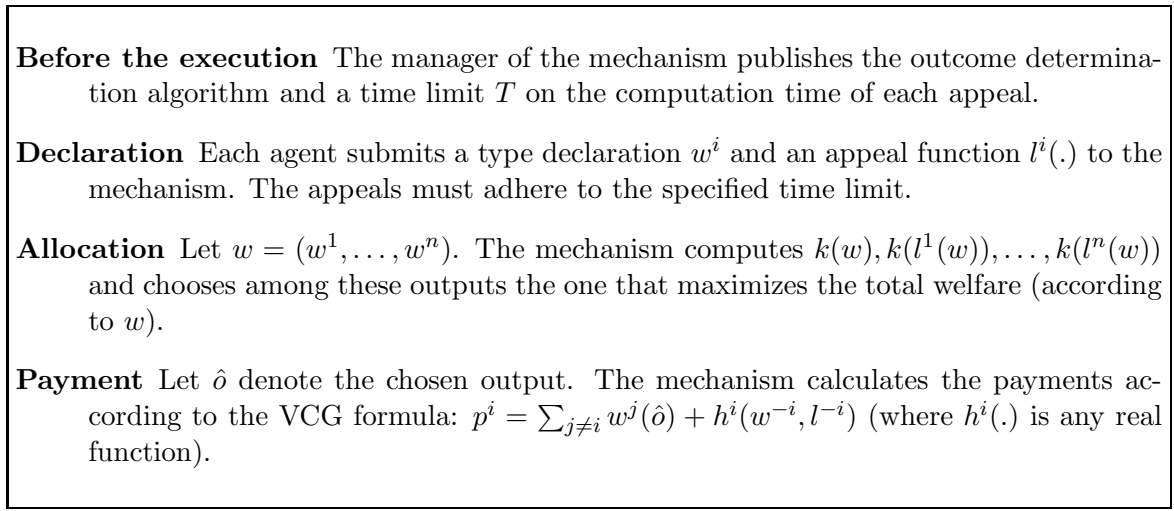

\label{fig:secChance} 
  \centering
  \begin{minipage}[h]{\textwidth}
\begin{center}
\fbox{\parbox{\textwidth}{

\begin{description}
    \item[Before the execution] The manager of the mechanism
    publishes the outcome determination algorithm and a time limit
    $T$ on the computation time of each appeal.

  \item[Declaration] Each agent submits a type declaration $w^i$ and an appeal function
$l^i(.)$ to the mechanism. The appeals must adhere to  the
specified time limit.

  \item[Allocation] Let $w = (w^1, \ldots, w^n)$.
The mechanism computes $k(w), k(l^1(w)), \ldots, k(l^n(w))$ and
chooses among these outputs the one that maximizes the total
welfare (according to $w$).

  \item[Payment] Let $\hat{o}$ denote the chosen output. The mechanism
calculates the payments according to the VCG formula: $p^i =
\sum_{j \neq i}w^j(\hat{o}) + h^i(w^{-i}, l^{-i})$ (where $h^i(.)$
is any real function).
\end{description}
}}
\end{center}
\end{minipage}
  \caption{The Second Chance Mechanism}
\end{figure}

\smallspace{\bf Remarks} The agents send programs that represent
their appeal functions to the mechanism. These programs are then
executed by the mechanism. The mechanism can terminate the
computation of each appeal after $T$ units of computation time
(and refer to the vector of declarations $w$ as if it is not in
the appeal's domain). Thus, we can assume w.l.o.g. that all
appeals adhere to the given time limit. A discussion on the choice
of the time limit and alternative representations of the appeal
functions appears in Section \ref{LtimeLimitChoice}. We believe
that it is possible to construct software tools and APIs that will
make the formulation of the appeals an easy task.

The functions $h^i(.)$ do not play any role in the agents'
considerations as every $h^i(.)$ is independent of $i$'s actions.
Until Section \ref{LobtainingIR} it is possible to simply assume
that $h^i(.) \equiv 0$ for all $i$. In Section \ref{LobtainingIR}
we will use these functions in order to satisfy individual
rationality.

\begin{definition}
\label{LtruthfulAct} {\bf (truthful action)} An {\em action} in
the second chance mechanism is a pair $(w^i, l^i)$ where $w^i$ is
a type declaration and $l^i(.)$ is an appeal function. An action
is called {\em truthful} if $w^i = v^i$.
\end{definition}

\noindent The following observation is a key property of the
mechanism.

\begin{proposition}\label{LsecChanceWelfare}
Consider a second chance mechanism with an output algorithm $k$.
For every type vector $v = (v^1, \ldots, v^n$), if all agents are
truth-telling, $g(v, \hat{o}) \geq g(w, k(v))$.
\end{proposition}
\qed

\smallspace In other words, when the agents are truth-telling, the
result of the mechanism is at least as good as $k(v)$. The proof
is immediate from the definition of the mechanism. We now
formulate an analog of Lemma \ref{LLemmaVCGutil}. The proof is
similar to the lemma's proof and is henceforth omitted.

\begin{lemma}\label{LLemmaSecChanceUtil}{\bf (second chance utility)} Consider
a second chance mechanism. Let $\hat{o}$ be the chosen output. The
{\em utility} of agent $i$ equals $g((v^i, w^{-i}), \hat{o}) +
h^i(w^{-i}, l^{-i})$.
\end{lemma}
\qed

\smallspace Therefore, informally, it is beneficial for an agent
to declare $w^i \neq v^i$ only if it either helps the output
algorithm $k(.)$ to compute a better result (w.r.t. $(v^i,
w^{-i})$) or it helps one of the appeals of the other agents.

Note that lying to a second chance mechanism may {\em harm} an
agent in two ways. First, it can damage the output algorithm
$k(.)$. Second, it can cause the mechanism to measure the welfare
according to a wrong type vector and thus cause it to choose an
inferior output.

\notation\ We say that a second chance mechanism is {\em
T-limited} if the time limit it specifies is $T$. Similarly, an
algorithm is called {\em T-limited} if its computational time
never exceeds $T$ units of computation.

The following proposition is obvious.

\begin{proposition} Consider a $T$-limited second chance
mechanism. If the output algorithm of the mechanism is also
$T$-limited, the overall computational time of the mechanism is
$O(nT)$.
\end{proposition}\qed

\subsubsection{A toy example}
Consider a combinatorial auction of two items.  A type of an agent
is a 3-tuple representing its value for every non empty subset of
items. Suppose that agent $i$ values the pair of items at $\$3$
million but values every single item at $\$1$ million. Its type
is, therefore, $v^i = \{ 3, 1, 1 \}$. Suppose that the agent
notices that the allocation algorithm often produces better
allocations if it declares $w^i = \{ 3, 0, 0 \}$ (i.e., it hides
its willingness to accept only one item). In a VCG-based mechanism
the agent may prefer to declare $w^i$ instead of its actual type.
This might cause two problems:

\begin{enumerate}
  \item Even when the others are truthful,
  there may be many type vectors $v^{-i}$ belonging to the other agents, for
  which declaring $w^i$ damages the chosen allocation, i.e.,
  $g((v^i, w^{-i}), k((w^i, w^{-i}))) < g((v^i, w^{-i}),k((v^i,
  w^{-i})))$.
  \item Even when this is not the case and every agent $i$ chooses a
  declaration $w^i$ such that $g((v^i, w^{-i}), k((w^i, v^{-i}))) \geq g((v^i, w^{-i}),
  k(w))$, it may be that according to the {\em actual} type vector $v$ the
  output $k(w)$ may be inferior to $k(v)$ (i.e., $g(v,k(w)) < g(v,k(v))$).
\end{enumerate}

\smallspace The second chance mechanism enables the agent to {\bf
check} whether declaring a falsified type would yield a better
result. Instead of declaring $w^i = \{ 3, 0, 0 \}$, the agent can
declare its actual type and define its appeal as $l^i(w') = (w^i,
w'^{-i})$. In this way the agent enjoys both worlds. In cases
where the falsified type is better, the mechanism will prefer
$k((w^i, w^{-i}))$ over $k((v^i, w^{-i}))$. In cases where the
truthful declaration is better, the mechanism will prefer $k((v^i,
w^{-i}))$. Note that the mechanism allows an appeal to modify not
only the declaration of the agent that submitted it but also the
whole vector of declarations. This will allow us to provide a
strong argument for truth-telling.

\paragraph{Possible Variants of the Second Chance Mechanism}

\smallspace One alternative definition of the mechanism is to let
the agents submit outcome determination algorithms instead of
appeals. It is possible to apply reasoning similar to ours to this
variant. However, formulating such output algorithms might be a
demanding task for many applications. There are also other more
delicate differences.

Another possibility is to define a multi-round variant of the
mechanism. In the first round the agents submit type declarations
$w$. Then, at each round, each agent gets the chance to improve
the allocation found by the algorithm $k(w)$. The mechanism
terminates when no agent improves the current allocation. The
strategy space of multi-round mechanisms is very complex. Yet,
under myopic behavior \cite{pa99}, arguments similar to ours can
be used to justify truthful behavior. Such arguments may explain
the relative success of ad hoc mechanisms such as iterative VCG
(IVG) and AUSM\footnote{Both mechanisms are in the spirit of our
second chance mechanism, as they let the agents improve the
allocation. The actual rules of these mechanisms are complicated
and are described by \shortciteR{BLP}.} reported by
\shortciteR{BLP}.

\paragraph{Standard Equilibria in Second Chance Mechanisms}
The second chance mechanism uses VCG payments and, therefore, the
theorems of the first part of the paper apply to it. From Lemma
\ref{LLemmaSecChanceUtil}, a vector of truthful actions is an ex
post equilibrium if and only if the resulting allocation $\hat{o}$
is optimal in the range of the algorithm. Moreover, consider agent
$i$ and let $(w^{-i}, l^{-i})$ be a set of actions of the other
agents. $(w^i, l^i)$ is a best response for the agent if and only
if the resulting allocation $\hat{o}$ is optimal in the range of
the underlying algorithm with respect to $(v^i, w^{-i})$. At least
intuitively, finding such a response is at least as hard as
finding an allocation that is optimal in the range of the
algorithm. Thus, one should not expect the agents to follow
equilibrium strategies in the traditional sense. We argue that
similar arguments can be made for every game in which computing
the best response is computationally difficult. Hence, an argument
that takes into account the agents' own limitations is required.
We note that we did not succeed in finding natural complexity
limitations under which truth-telling is an equilibria for the
agents. We leave this as an intriguing open problem.

\subsection{The Rationale for Truth-telling}
As we noted, standard equilibria should not be expected in second
chance mechanisms. This section formulates the rationale for
truth-telling under these mechanisms. We first introduce the
notion of feasibly dominant actions\footnote{We make a standard
distinction between an action and a strategy -- a mapping from the
agent's type to its action.} which takes into account the fact
that the agents' capabilities are limited. We then demonstrate
that under reasonable assumptions about the agents, truthful,
polynomial time, feasibly dominant actions exist.

\subsubsection{Feasible Truthfulness}

The basic models of equilibria in game theory are justified by the
implicit assumption that the agents are capable of computing their
best response functions. In many games, however, the action space
is huge and this function is too complex to be computed, even
approximately within a reasonable amount of time. In such
situations the above assumption seems no longer valid.

In this section we re-formulate the concept of dominant actions
under the assumption that agents have a limited capability of
computing their best response. Our concept is meant to be used in
the context of one stage games, i.e. games in which the agents
choose their actions without knowing anything about the others'
choices. The second chance mechanism is a one stage-game. In a
nutshell, an action is feasibly dominant if the agent is not {\em
aware} of {\em any} situation (a vector of the other agents'
actions) where another action is better for it.

\notation\ We denote the action space of each agent $i$ by $A^i$.
Given a tuple $a = (a^1, \ldots, a^n)$ of actions chosen by the
agents, we denote the utility of agent $i$ by $u^i(a)$.

\begin{definition}
\label{strategicknowledge} {\bf (revision function)} A {\em
revision function} of agent $i$ is any partial function of the
form $b^i:A^{-i} \rightarrow A^i$.
\end{definition}

\smallspace The semantics of $b^i(a^{-i})$ is ``If I knew that the
actions of the others are $a^{-i}$, I would choose $b^i(a^{-i})$
(instead of $a^i$)''. A revision function captures all the cases
where the agent knows how it would like to act if it knew the
others' actions. Note that optimal revision functions are standard
best-response functions. When a vector of actions $a^{-i}$ does
not belong to the domain of $b^i(.)$, the semantics is that the
agent prefers to stick to its action.

\begin{definition}
\label{feasiblebestresponse} {\bf (feasible non-regret)} Let $i$
be an agent, $b^i(.)$ its revision function, and $a^{-i}$ a vector
of actions for the other agents. An action $a^i$ satisfies the
{\em feasible non-regret} condition (w.r.t. $a^{-i}$ and $b^i$),
if either $a^{-i}$ is not in the domain of $b^i$ or
$u^i((b^i(a^{-i}), a^{-i})) \leq u^i(a)$.
\end{definition}

\smallspace In other words, other actions may be better against
$a^{-i}$, but the agent is unaware of them or cannot compute them
when choosing its action.

When the revision function of the agent is optimal, a feasible
non-regret is equivalent to the standard non-regret (best
response) condition.

\begin{definition}
\label{feasiblydominant} {\bf (feasibly dominant action)} Let $i$
be an agent, $b^i(.)$ its revision function. An action $a^i$ is
called {\em feasibly dominant} (w.r.t. $b^i(.)$) if for {\bf
every} vector $a^{-i}$ of the actions of the other agents, $a^i$
satisfies the feasible non-regret condition (w.r.t. $a^{-i}$ and
$b^i$).
\end{definition}

\smallspace Put differently, an action $a^{i}$ is feasibly
dominant if (when choosing its action) the agent is not aware of
any action $a'^{i}$ and any vector $a^{-i}$ of the actions of the
other agents, such that it is better off choosing $a'^{i}$ when
the others choose $a^{-i}$. A dominant action is always feasibly
dominant. When the revision function is optimal, a feasibly
dominant action is dominant.

\paragraph{Example}
In order to demonstrate the concept of feasibly dominant actions
consider a chess match in which Alice and Bob submit computer
programs that play on their behalf. Currently, of course it is not
known how to compute an equilibrium in chess and therefore
standard equilibria are not relevant for the analysis of such a
game. A program $a_A$ is feasibly dominant for Alice if she is not
aware of any possible program of Bob against which she is better
off submitting another program. \smallspace

\begin{definition}
\label{LfeasiblytrutfulAction} {\bf (feasibly truthful action)} An
action $a^i$ in the second chance mechanism is called {\em
feasibly truthful} if it is both, truthful and feasibly dominant.
\end{definition}

\subsubsection{Natural revision functions that give rise to feasibly truthful
actions}\label{LnatRevFunc}
Beforehand we showed that when the agents are truthful, the total
welfare is at least $g(v,k(v))$. We also argued that if a feasibly
truthful action is available, the agent has a strong incentive to
choose it. This subsection demonstrates that under reasonable
assumptions about the agents, polynomial time feasibly truthful
actions do exist.

\notation\ We let $\bot$ denote the empty appeal. By $(w, \bot)$
we denote an action vector where the declaration of each agent $i$
is $w^i$ and all the appeals are empty.

\begin{definition} {\bf (appeal-independent revision function)} A revision
function $b^i(.)$ is called {\em appeal independent} if every
vector in its domain includes only empty appeals, i.e. for all
$a^{-i} \in dom(b^i)$, there exists a vector $w^{-i}$ such that
$a^{-i} = (w^{-i}, \bot)$.
\end{definition}

We say that an appeal independent function is {\em $T$-limited} if
its own computational time is bounded by $T$ and so is every
appeal function in its range.

The class of appeal-independent revision functions represents
agents that only explore the output algorithm (or alternatively,
base their choice of action solely on the output algorithm). This
approach seems reasonable as the space of appeals of the other
agents is huge, with no apparent structure. At least intuitively,
it seems unreasonable that an agent will be able to lie in a way
that will improve the result of the appeals of the other agents
with significant probability. Moreover, as we commented, an agent
has an obvious potential loss from misreporting its type.

\begin{theorem}\label{appealIndAction}
Consider a second chance mechanism with a $T$-limited output
algorithm. Suppose that an agent has a $T$-limited
appeal-independent revision function. For every $T' = \Omega(T)$,
if the mechanism is $T'$-limited, the agent has a feasibly
truthful action.
\end{theorem}
\begin{proof} Let $b^i(.)$ be the agent's revision function.
Define an appeal $l^i(.)$ as follows. For every vector $w^{-i}$,
let $(w^i, \tau^i) = b^i((w^{-i}, \bot))$. Let $w = (w^i,
w^{-i})$. Consider the outputs $o_1 = k(w)$ and $o_2 =
k(\tau^i(w))$. We define $l^i(w)$ to be the {\em better} of the
two outputs, i.e., $l^i(w) = \arg\max_{j = 1,2} \, g((v^i,
w^{-i}), o_j)$. Intuitively, $l^i(.)$ checks whether declaring
$w^i$ is helpful for the agent.

\begin{claim} $a^i = (v^i, l^i)$ is feasibly truthful.
\end{claim}
\begin{proof} If it is not, there exists a vector $a^{-i} = (w^{-i},
\bot)$ in the domain of $b^i(.)$ such that $u(a^i, a^{-i}) <
u(b^i(a^{-i}), a^{-i})$. Let $b^i(a^{-i}) = (w^i, \tau^i)$. Recall
that according to Lemma \ref{LLemmaSecChanceUtil}, the agent's
utility is equivalent to the total welfare $g((v^i, w^{-i}), o)$
of the chosen output $o$ (up to adding $h^i(.)$, which is
independent of the agent's actions).

Consider the case when the agent's action is $b^i(a^{-i})$. Let
$\hat{o}$ denote the chosen output in this case. According to the
definition of the mechanism, $\hat{o}$ is taken from the set
$\{o_1, o_2\}$ and the welfare is measured according to the
declaration $w$.

When the agent chooses the truthful action $a^i$, the output
(denoted $\tilde{o}$) is chosen from the outputs $o_0 = k((v^i,
w^{-i}))$ (from the definition of the mechanism), and both, $o_1,
o_2$ (from the definition of $l^i$). This is a {\bf superset} of
the set outputs of the first case. Moreover, the output is chosen
according to the ``right'' type vector $(v^i, w^{-i})$. Thus,
$g((v^i, w^{-i}), \tilde{o}) \geq g((v^i, w^{-i}), \hat{o})$,
implying that the agent has a higher utility in the second case --
a contradiction.
\end{proof}

\smallspace It remains to show that $l^i(.)$ is
$\Omega(T)$-limited. This is obvious as both, $k(.)$ and
$\tau^i(.)$ are $T$-limited. This completes the proof of the
theorem.
\end{proof}

\smallspace Given the agent's revision function, it is easy to
construct the appeal $l^i(.)$ defined above (i.e., construct the
program that computes it). Thus, if the agent has such an appeal
independent function, it can guarantee {\bf itself} a feasibly
dominant action.

A more general class of revision functions can be found in the
Appendix. Interestingly, there is a tradeoff between the
generality of the class and the time limit, which suffices for
feasible truthfulness.

\subsection{Remarks on the Choice of the Time Limit}\label{LtimeLimitChoice}

Sections \ref{LnatRevFunc} and \ref{boundedRevisionFD} demonstrate
two natural classes of revision functions under which the agents
have polynomial time feasibly truthful actions. We do not claim
that every revision function in practice will fall into these
categories. Yet, it is plausible that this will be the case in
many applications. In general, there exists a tradeoff between the
generality of the class of the revision functions and the time
limit required for feasible truthfulness. In particular, without
any time limit, submitting an optimal appeal is dominant. On the
other hand, it is plausible that small time limits will suffice in
practice. We leave a more comprehensive study of this tradeoff to
future research.

An interesting future direction is to develop representations of
the appeal functions that {\bf relate} the time limit imposed on
each agent to its actual revision function. One possibility is to
represent the appeals by decision trees where the agents are
required to supply for each leaf $\alpha$, a type vector
$v_{\alpha}$, such that the algorithm's result is strictly
improved when it is given $l(t_{\alpha})$ instead of the actual
input $v_{\alpha}$. $v_{\alpha}$ proves to the mechanism that the
computational time required to compute the leaf $\alpha$ is indeed
needed in order to represent the agent's revision function. A
related possibility is to allow the agent to purchase additional
computational time.

Currently, we do not know whether every polynomial class of
revision functions guarantees the existence of polynomial feasibly
truthful actions. If an agent has substantial knowledge of the
appeal space of the other agents, it may be able to find a
falsified declaration that causes ``typical'' appeals to produce
better results. In such a case, it may be beneficial for an agent
to lie. We do not know whether such knowledge will exist in
practice. If yes, it may be possible to overcome this by allowing
the agents to submit meta-appeals, i.e., functions that let the
agents modify the input of the appeals of the other agents. We
leave this to future research.

\subsection{Obtaining Individual Rationality}\label{LobtainingIR}
\label{partic} A basic desirable property of mechanisms is that
the utility of a truthful agent is guaranteed to be non-negative
(individual rationality). In this section we construct a variant
of second chance mechanisms that satisfies this property.

Let $g_{opt}(v)$ denote the optimal welfare obtained when the type
vector is $v$ . We shall assume that for each agent $i$, there
exists a type ${\underline{v}}^i$ such that for every $v = (v^1,
\ldots, v^n)$, $g_{opt}((\underline{v}^i, v^{-i})) \leq
g_{opt}(v)$. We call such a type the {\em lowest}. In a
combinatorial auction for example, the lowest type is defined by
the zero valuation $\underline{v}^i(s) = 0$ for every combination
$s$ of items.

The Clarke mechanism \cite{Cla71} is a VCG mechanism in which
$h^i(w^{-i}) = - g_{opt}(\underline{v}^i, w^{-i})$, i.e., $p^i(w)
= \sum_{j \neq i} w^j(opt(w)) - g_{opt}(\underline{v}^i,
(w^{-i}))$. In other words, each agent pays the welfare loss it
causes to the society. Thus, it is natural to define the payment
of a VCG-based mechanism as $ \sum_{j \neq i} w^j(opt(w)) -
g((\underline{v}^i, w^{-i}), k((\underline{v}^i, w^{-i})))$.

Like truthfulness, individual rationality may not be preserved
when the optimal algorithm in the Clarke mechanism is replaced by
a sub-optimal one. In order to fix this we need to ensure that the
result of the algorithm will not improve when the declaration
$w^i$ is replaced by the lowest type $\underline{v}^i$.

\begin{definition}{\bf (lowest type closure)} Given an allocation algorithm $k(w)$ we define
its {\em lowest type closure} $\tilde{k}$ as the best allocation
(according to $w$) among  the outputs $(k(w), k((\underline{v}^1,
w^{-1})), \ldots, k((\underline{v}^n, w^{-n})))$.
\end{definition}

\smallspace Since $\tilde{k}(.)$ calls $k(.)$ $n$ times, if $k$ is
$T$-limited, then $\tilde{k}$ is $O(nT)$-limited.

\begin{claim} For every $w$, $g(w, \tilde{k}(w)) \geq g((\underline{v}^i, w^{-i}), k((\underline{v}^i,
w^{-i})))$.
\end{claim}
\begin{proof} Since $k((\underline{v}^i,
w^{-i}))$ is a candidate output that $\tilde{k}$ tests, $g(w,
\tilde{k}(w)) \geq g(w, k((\underline{v}^i, w^{-i})))$. Given the
definition of $\underline{v}^i$, $g(w, k((\underline{v}^i,
w^{-i}))) \geq g((\underline{v}^i, w^{-i}), k(\underline{v}^i,
w^{-i}))$, the claim follows.
\end{proof}

\begin{definition}{\bf (second chance-IR)} Given an allocation algorithm $k(w)$
and a time limit $T$ we define the corresponding {\em second
chance-IR mechanism} as the second chance mechanism with output
algorithm $\tilde{k}(.)$, time limit $T$, and for every agent $i$,
$h^i(w^{-i}) =  - g((\underline{v}^i, w^{-i}), k((\underline{v}^i,
w^{-i})))$.
\end{definition}

\smallspace The utility of a truthful agent in the above mechanism
equals $u^i = g(w, \hat{o}) - g((\underline{v}^i, w^{-i}),
k((\underline{v}^i, w^{-i}))) \geq g(w, k(w)) -
g((\underline{v}^i, w^{-i}), k((\underline{v}^i, w^{-i}))) \geq
0$. Therefore, the mechanism satisfies individual rationality.

\section{Conclusion and Future Research} \label{epilog}
This paper studies VCG mechanisms in which the optimal outcome
determination algorithm is replaced by some sub-optimal but
computationally tractable algorithm. The first part of the paper
shows that for a wide range of problems, such mechanisms lose the
game theoretic virtues of their optimal counterparts. Similar
results hold for affine maximization. These results do not leave
much hope for the development of polynomial time truthful
mechanisms for many problems of high complexity.

The second part of the paper proposes a general method for
overcoming the difficulty of constructing truthful mechanisms.
Given any algorithm for the underlying optimization problem we
define the second chance mechanism based on it. We demonstrate
that under reasonable assumptions about the agents, truth-telling
is still the rational strategy for the agents. When the agents are
truthful, the welfare obtained by the mechanism is at least as
good as the one obtained by the underlying algorithm.

Successful implementation of second chance mechanisms relies on
several tools to be developed -- in particular, tools for the
description of valuations and appeal functions. These
``engineering'' issues require further exploration.

It is important to stress that the second chance method has not
yet been tested. In particular, the truthfulness of the agents
should be validated experimentally. On the other hand, we believe
that in practice, small time limits on the agents' appeals are
likely to guarantee the truthfulness of the agents. Several
questions regarding the payment properties of second chance
mechanisms are open. We leave them for future research.

Several open questions, which directly stem from this work, are
raised within the body of the paper.

\acks{We thank Abraham Newman and Motty Perry for helpful
discussions at various stages of this work. We thank Ron Lavi,
Ahuva Mu'alem, Elan Pavlov, Inbal Ronen, and the anonymous
reviewers for comments on earlier drafts of this paper. Noam Nisan
was supported by grants from the Israel Science Foundation and
from the USA-Israel Binational Science Foundation. Amir Ronen was
supported in part by grant number $969/06$ from the Israel Science
Foundation. A preliminary version of this paper appeared in the
proceedings of the 3rd ACM Conference on Electronic Commerce (EC'
01).}

\appendix

\section{$d$-bounded Revision Functions}
The class of $d$-bounded revision functions represents agents
that, in addition to the output algorithm, explore a polynomial
family of potential appeals of the other agents. This class is a
generalization of $d$-limited appeal-independent functions.

\begin{definition} {\bf ($d$-bounded revision function)}
We say that a revision function $b^i(.)$ is {\em $d$-bounded} if
the following hold:
\begin{enumerate}
  \item The revision function $b^i(.)$ is $O(n^d)$-limited.
  \item Let
  \[ L = \{l_j \, | \, \exists l^{-i,-j}, w^{-i} s.t. \, (w^{-i},
    (l^i,l^{-i,-j}))\in Dom(b^i)\}
 \;  \bigcup \;
  \{l_i \, | \, \exists (w^{-i}, l^{-i}), w^i s.t. \, (w^i,
  l^i) = b^i((w^{-i}, l^{-i})) \} \] be the family of all appeals that appear in
  either the domain or range of $b^i(.)$. Then $|L| = O(n^d)$.
  \item There exists a constant $c$ such that every appeal $l \in L$ is $c n^d$-limited.
\end{enumerate}
\end{definition}

\begin{theorem}\label{boundedRevisionFD}
Consider a second chance mechanism with an $O(n^d)$-limited output
algorithm. Suppose that an agent has a $d$-bounded revision
function. For every $T' = \Omega(n^{2d})$, if the mechanism is
$T'$-limited, the agent has a feasibly truthful action.
\end{theorem}
\begin{proof} Let $i$ be the agent and let $b^i$ be its revision function.
We again use a simulation argument in order to define the appeal
$l^i(.)$. For every vector $w^{-i}$ we compute the following
outputs:
\begin{enumerate}
  \item $o_0 = k(w)$.
  \item Similarly to the the proof of Theorem \ref{appealIndAction}, let $L = \{ \tau_1 \ldots \tau_{|L|} \}$
  be the family of all appeal functions that are in the domain or the
  range of $b^i$. For all $j = 1, \ldots |L|$ define $o_j = k(\tau_j(w)))$.
  \item Define $l(w) = \arg\max_{0 \geq j \geq |L|} \, g((v^{i}, w^{-i}),
  o_j)$ as the output with the maximum welfare according to $(v^{i},
  w^{-i})$ among all the outputs defined above.
\end{enumerate}

\begin{claim} $l^i(.)$ is $n^{2d}$-limited.
\end{claim}
\begin{proof} W.l.o.g. the running time of $k(.)$ is bounded by $cn^d$. Otherwise, we will raise the constant.
According to the definitions, the appeal $l^i$ performs $n^d + 1$
computations, each requiring at most $cn^d$ time units. Thus, the
overall computation takes at most $O(n^{2d})$.
\end{proof}

\begin{claim} $a^i = (v^i, l^i)$ is feasibly truthful.
\end{claim}
\begin{proof} Assume by contradiction that there exists an action
vector $a^{-i}$ in $dom(b^i)$ such that $u((a^i, a^{-i}) <
u((b^i(a^{-i}), a^{-i})$.

Consider the case when the agent chooses $b^i(a^{-i}) = (w^i,
\tau^i)$. The mechanism takes the output $\hat{o}$ that maximizes
the welfare (according to $w$) from the following set $S$ of
outputs:

\begin{enumerate}
  \item $o_0 = k(w)$.
  \item $o_j = k(l_j(w))$ for every $j \neq i$, i.e. the result of
  the appeals of the other agents.
  \item $o_i = k(\tau^i(w))$.
\end{enumerate}

When the agent chooses $a^i$, the outputs are measured according
to the ``right'' type vector $(v_i, w^{-i})$. Moreover, it is
taken from the following {\bf superset} of the outputs in $S$:
\begin{enumerate}
  \item $o'_0 = k((v_i, w^{-i}))$ (from the definition of the mechanism).
  \item $o'_j = k(l_j((v_i, w^{-i})))$ for every $j \neq i$, i.e., the result of
  the appeals of the other agents (also, from the definition of the
  mechanism).
  \item $o'_j = k(\tau(w))$ for every $\tau \in L$. Since $a^{-i}$
  is in the domain of $b^i$, this set includes {\bf all} the
  outputs of the form $k(l_j(w))$ from the case where $i$ chooses
  $b^i(a^{-i})$. It also contains the result of its own appeal $\tau^i(w)$.
  \item $k(w)$ (from the definition of $l^i(.)$).
\end{enumerate}

\smallspace Let $\tilde{o}$ be the chosen output in this case.
Since the set of outputs in the second case is a superset of the
first, $g((v_i, w^{-i}), \tilde{o}) \geq g((v_i, w^{-i}),
\hat{o})$. According to Lemma \ref{LLemmaSecChanceUtil} the
utility of the agent when choosing $a^i$ is thus higher than when
choosing $b^i(a^{-i})$ -- a contradiction.
\end{proof}

\smallspace This completes the proof of Theorem
\ref{boundedRevisionFD}.
\end{proof}

\smallspace As in the case of appeal-independent functions, the
theorem gives a prescription for constructing an appeal that
guarantees the agent a feasibly dominant action.

\vskip 0.2in
\bibliographystyle{theapa}
\bibliography{bib12Compfeas}

\end{document}